\documentclass[letterpaper,11pt]{article}

\usepackage[margin=1.1in]{geometry}
\usepackage{amsmath,amssymb,amsthm}
\usepackage{cite}

\newcommand{\itemspacing}{\setlength{\itemsep}{0pt}}

\renewcommand{\geq}{\geqslant}
\renewcommand{\leq}{\leqslant}

\newcommand{\Bool}{\ensuremath{\{0,1\}}}

\newcommand{\imp}{\rightarrow}

\newcommand{\OR}{\ensuremath{\mathrm{OR}}}
\newcommand{\NAND}{\ensuremath{\mathrm{NAND}}}
\newcommand{\ORconj}{\ensuremath{\OR{}\text{-conj}}}
\newcommand{\NANDconj}{\ensuremath{\NAND{}\text{-conj}}}
\newcommand{\NANDconjLinebreak}{NAND-conj} 
\newcommand{\IMconj}{\ensuremath{\text{IM-conj}}}

\newcommand{\bwcomp}[1]{{\widetilde{#1}}}
\newcommand{\proj}[1]{{\mathrm{proj}_{#1}}}

\newcommand{\powerset}[1]{{\mathcal{P}({#1})}}

\newcommand{\pppleq}{\leq_{\mathrm{ppp}}}

\newcommand{\width}[1]{{\mathrm{width}\,({#1})}}
\newcommand{\vrank}[1]{{\mathrm{vrank}\,({#1})}}

\newcommand{\Ror}{R_{\mathrm{OR}}}
\newcommand{\Rnand}{R_{\mathrm{NAND}}}
\newcommand{\Rork}[1][k]{R_{\mathrm{OR}, {#1}}}
\newcommand{\Rnandk}[1][k]{R_{\mathrm{NAND}, {#1}}}
\newcommand{\Rimp}{R_{\imp}}
\newcommand{\Rpmi}{R_{\leftarrow}}
\newcommand{\Rneq}{R_{\neq}}
\newcommand{\Req}{R_{=}}
\newcommand{\Reqk}[1][k]{R_{{=},{#1}}}
\newcommand{\Rzero}{R_{\mathrm{zero}}}
\newcommand{\Rone}{R_{\mathrm{one}}}
\newcommand{\GammaPin}{\ensuremath{\Gamma_\mathrm{\!pin}}}

\newcommand{\tuple}[1]{\left\langle\,{#1}\,\right\rangle}

\newcommand{\abar}{\bar{a}}
\newcommand{\bbar}{\bar{b}}
\newcommand{\cbar}{\bar{c}}
\newcommand{\dbar}{\bar{d}}
\newcommand{\vbar}{\bar{v}}
\newcommand{\xbar}{\bar{x}}

\newcommand{\numBIS}{\ensuremath{\#\mathrm{BIS}}}
\newcommand{\numSAT}{\ensuremath{\#\mathrm{SAT}}}
\newcommand{\numwHIS}{\ensuremath{\#w\mathrm{\text{-}HIS}}}
\newcommand{\numwHISd}[1][d]{\ensuremath{\#w\mathrm{\text{-}HIS}_{#1}}}

\newcommand{\CSP}{\ensuremath{\mathrm{CSP}}}
\newcommand{\numCSP}{\ensuremath{\#\mathrm{CSP}}}
\newcommand{\numCSPd}[1][d]{\ensuremath{\numCSP_{#1}}}

\newcommand{\Ptime}{\ensuremath{\mathbf{P}}}
\newcommand{\numP}{\ensuremath{\#\mathbf{P}}}
\newcommand{\FPtime}{\ensuremath{\mathbf{FP}}}
\newcommand{\NPtime}{\ensuremath{\mathbf{NP}}}
\newcommand{\RPtime}{\ensuremath{\mathbf{RP}}}

\newcommand{\APredto}{\leq_{\mathrm{AP}}}
\newcommand{\APred}{\leq_{\mathrm{AP}}}
\newcommand{\APequiv}{\equiv_{\mathrm{AP}}}

\newcommand{\GFtwo}{\ensuremath{\mathrm{GF}_2}}

\newcommand{\Nesetril}{Ne\v{s}et\v{r}il}

\theoremstyle{plain}
\newtheorem{theorem}{Theorem}[]
\newtheorem{lemma}[theorem]{Lemma}
\newtheorem{corollary}[theorem]{Corollary}
\newtheorem{proposition}[theorem]{Proposition}

\theoremstyle{definition}

\newtheorem*{proviso*}{Proviso}

\title {The Complexity of Approximating\\
        Bounded-Degree Boolean \#CSP\thanks{The work described in this paper was partly supported by EPSRC Research Grant
(refs\ EP/I011528/1  and EP/I012087/1) ``Computational Counting''$\!$.}}

\author{%
    \ Martin Dyer,%
    \thanks{~School of Computing, University of Leeds, Leeds, LS2~9JT, UK.}
    \ Leslie Ann Goldberg,%
    \thanks{~Dept of Computer Science, University of Liverpool,
        Liverpool, L69~3BX, UK.}
    \\Markus Jalsenius%
    \thanks{~Dept of Computer Science, University of Bristol,
        Merchant Venturers Building, Bristol, BS8~1UB, UK.}
    \ and David Richerby%
    \footnotemark[3]
}

\date{}

\begin{document}
\maketitle

\begin{abstract}
    The degree of a CSP instance is the maximum number of times that
    any variable appears in the scopes of constraints. We consider the
    approximate counting problem for Boolean CSP with bounded-degree
    instances, for constraint languages containing the two unary
    constant relations $\{0\}$ and $\{1\}$. When the maximum allowed
    degree is large enough (at least $6$) we obtain a complete
    classification of the complexity of this problem.  It is exactly
    solvable in polynomial-time if every relation in the constraint
    language is affine. It is equivalent to the problem of
    approximately counting independent sets in bipartite graphs if
    every relation can be expressed as conjunctions of $\{0\}$,
    $\{1\}$ and binary implication.  Otherwise, there is no FPRAS
    unless $\NPtime = \RPtime$. For lower degree bounds, additional
    cases arise, where the complexity is related to the complexity
    of approximately counting independent sets in hypergraphs.
\end{abstract}


\section{Introduction}
\label{sec:Intro}

In the constraint satisfaction problem (CSP), we seek to assign values
from some domain to a set of variables, while satisfying given
constraints on the combinations of values that certain tuples of the
variables may take.  Constraint satisfaction problems are ubiquitous
in computer science, with close connections to graph theory, database
query evaluation, type inference, satisfiability, scheduling and
artificial intelligence \cite{KV2000:Conjunctive,
Kum1992:CSP-algorithms, Mon1974:Constraints}.  CSP can also be
reformulated in terms of homomorphisms between relational structures
\cite{FV1998:MMSNP} and conjunctive query containment in database
theory \cite{KV2000:Conjunctive}.  Weighted versions of CSP
appear in statistical physics, where they correspond to partition
functions of spin systems \cite{Wel1993:Complexity}.

We give formal definitions in Section~\ref{sec:Prelim} but, for now,
consider an undirected graph $G$ and the CSP where the domain is
$\{\mathrm{red}, \mathrm{green}, \mathrm{blue}\}$, the variables are
the vertices of $G$ and the constraints specify that, for every edge
$xy\in G$, $x$ and $y$ must be assigned different values.  Thus, in a
satisfying assignment, no two adjacent vertices are given the same
colour: the CSP is satisfiable if, and only if, the graph is
3-colourable.  As a second example, given a formula in 3-CNF, we can
write a system of constraints over the variables, with domain
$\{\mathrm{true}, \mathrm{false}\}$, that requires the assignment to
each clause of the formula to satisfy at least one literal.
Clearly, the resulting CSP is directly equivalent to the original
satisfiability problem.


\subsection{Decision CSP}
\label{sec:Decision}

In the \emph{uniform constraint satisfaction problem}, we are given
the set of constraints explicitly, as lists of allowable combinations
for given tuples of the variables; these lists can be considered as
relations over the domain.  Since it includes problems such as
3-\textsc{sat} and 3-\textsc{colourability}, uniform CSP is
\NPtime{}-complete.  However, uniform CSP also includes problems in
\Ptime{}, such as 2-\textsc{sat} and 2-\textsc{colourability}, raising
the natural question of what restrictions lead to tractable problems.
It is natural to restrict either the form of the constraints or of the
instances.

The most common restriction is to allow only certain fixed
relations in the constraints.  The list of allowed relations is
known as the \emph{constraint language} and we write $\CSP(\Gamma)$
for the so-called \emph{non-uniform} CSP in which each constraint
states that the values assigned to some tuple of variables must be a
tuple in a specified relation in $\Gamma$.

The classic example of this is due to Schaefer \cite{Sch1978:Boolean}.
Restricting to Boolean constraint languages (i.e., those with domain
$\{0,1\}$), he showed that $\CSP(\Gamma)$ is in $\Ptime$ if $\Gamma$ is
included in one of six classes and is \NPtime{}-complete, otherwise.
The Boolean case of CSP is often referred to as ``generalized
satisfiability'' in the literature.  More recently, Bulatov has
produced a corresponding dichotomy for three-element domains
\cite{Bul2006:Ternary}.

Restricting to relations of fixed arity over arbitrary finite domains
has also been studied in depth. In particular, requiring $\Gamma$ to be a
single binary relation gives the directed graph homomorphism problem,
and the undirected graph homomorphism problem if the relation is also
required to be symmetric.  Hell and \Nesetril{} have shown that, for
every symmetric binary relation $E$, $\CSP(E)$ is either in \Ptime{}
or is \NPtime{}-complete~\cite{HN1990:H-color}.  They conjecture that
this holds for all binary relations.

In all the above cases, $\CSP(\Gamma)$ has been either in \Ptime{} or
\NPtime{}-complete and Feder and Vardi have conjectured that this
holds for all $\Gamma$~\cite{FV1998:MMSNP}.  No such dichotomy can
exist for the whole
of \NPtime{} because Ladner has shown that either $\Ptime{} =
\NPtime{}$ or there is an infinite, strict hierarchy between the two
\cite{Lad1975:Reducibility}.  However, a dichotomy for CSP is possible
as there are problems in \NPtime{}, such as graph Hamiltonicity and
even connectedness, that cannot be expressed as
$\CSP(\Gamma)$\footnote{This follows from the observation that any set
  $S$ of structures (e.g., graphs) that is definable in CSP has the
  property that, if $A\in S$ and there is a homomorphism $B\to A$,
  then $B\in S$; neither the set of Hamiltonian nor connected graphs
  has this property.} and Ladner's diagonalization does not seem to be
expressible in CSP \cite{FV1998:MMSNP}.  Resolving Hell and
\Nesetril{}'s conjecture for a class of simple acyclic digraphs would
immediately resolve the CSP dichotomy \cite{FV1998:MMSNP}, though
recent work on the dichotomy has focused on methods from universal
algebra --- see, for example, \cite{Bul2006:Ternary,
CJ2006:Constraints} and the references there.

Allowing arbitrary constraint languages but restricting the form of
the instances has also been studied. Dechter and Pearl
\cite{DP1989:Clustering} and Freuder \cite{Fre1990:CSP-tw} have shown
that even uniform CSP is in \Ptime{} on instances of bounded tree
width; see also \cite{KV2000:Games-CSP}.  Bounded tree width and other
similar restrictions are generalized by the ``guarded decompositions''
of Cohen, Jeavons and Gyssens~\cite{CJG2008:Guarded}.  Restricting the degree of
instances (the maximum number of times that each variable may appear
in the scopes of constraints) is incomparable but not much is known in
this case.  In the non-uniform Boolean case, Dalmau and Ford have
shown that, as long as $\Gamma$ contains the relations~$\Rzero=\{0\}$
and $\Rone=\{1\}$, $\CSP(\Gamma)$ for instances of degree at most
three has the same complexity as the case with no degree restrictions
\cite{DF2003:bdeg-gensat}. The degree-two case has not yet been
completely classified, though it is known that degree-2 $\CSP(\Gamma)$ is as
hard as general $\CSP(\Gamma)$ whenever $\Gamma$ contains $\Rzero$ and
$\Rone$ and some relation that is not a $\Delta$-matroid
\cite{Fed2001:Fanout,DF2003:bdeg-gensat}.


\subsection{Counting CSP}

A generalization of the classical constraint satisfaction problem is
to ask how many satisfying solutions there are, rather than just
whether the constraints are satisfiable.  This is referred to as the
counting CSP problem, \numCSP{}.  Clearly, the decision problem is
reducible to counting: if we can efficiently count the
solutions, we can efficiently determine whether there is at least one.
However, the converse does not hold: for example, there are well-known
polynomial-time algorithms that determine whether a graph
admits a perfect matching but it is \numP{}-complete to count the
perfect matchings, even in a bipartite graph
\cite{Val1979:Enumeration}.

The class \numP{} can be considered to be the counting analogue of
\NPtime{}: it is defined as the class of functions $f$ for which
there is a nondeterministic, polynomial-time Turing machine that has
exactly $f(x)$ accepting paths for every input $x$
\cite{Val1979:Permanent}.  The counting version of any \NPtime{}
decision problem is easily seen to be in \numP{}.  Note that, although
\numP{} plays a similar role in the complexity of function problems to
that of \NPtime{} in decision problems, problems that are complete for
\numP{} under appropriate reductions are, under standard
complexity-theoretic assumptions, considerably harder than
\NPtime{}-complete problems.  Toda has shown that $\Ptime^{\numP}$
includes the whole of the polynomial hierarchy \cite{Tod1989:PH},
but $\Ptime^{\NPtime}$ is generally thought not to.

Although it is not known if there is a dichotomy for CSP, Bulatov has
recently shown that, for every $\Gamma\!$, $\numCSP(\Gamma)$ is either
computable in polynomial time or \numP{}-complete
\cite{Bul2008:Dichotomy}.  Two of the present authors have since given
an elementary proof of this result and also shown the dichotomy to be
decidable~\cite{DR2011:dichotomy}.
However, it is not obvious how the methods of these results could
be applied to bounded-degree \numCSP{}.

So, although there is a full dichotomy for $\numCSP(\Gamma)$, results
for restricted forms of constraint language are still of interest.
For Boolean constraint languages, Creignou and Hermann have shown that
only one of Schaefer's polynomial-time cases survives the transition
to counting: $\numCSP(\Gamma)$ has a polynomial time algorithm if
every relation in $\Gamma$ is affine (i.e., the solution set of a
system of linear equations over \GFtwo{}) and is \numP{}-complete,
otherwise \cite{CH1996:Bool-numCSP}.  It is not surprising that there
are fewer tractable cases --- it is easy to arrange that every
instance of $\CSP(\Gamma)$ be trivially satisfiable (say, by
making the all-zeroes assignment satisfying), but the number of non-trivial
solutions might be difficult to compute.  Dyer, Goldberg and
Jerrum~\cite{DGJ2009:WBool} extended Creignou and Hermann's result to
weighted Boolean $\numCSP$. Cai, Lu and
Xia~\cite{CLXxxxx:Complex-numCSP,CLX2009:Holant-numCSP}
extended further to the case of complex weights and show that the
dichotomy holds for the restriction of the problem in which instances
have degree~$3$. Their result implies that the degree-3 problem
$\numCSP_3(\Gamma)$ ($\numCSP(\Gamma)$ restricted to instances of
degree~3) has a polynomial time algorithm if every relation in
$\Gamma$ is affine and is \numP{}-complete, otherwise.

The case where $\Gamma$ contains a single symmetric, binary relation
$E$ corresponds exactly to the problem of counting the homomorphisms
from an input graph to some fixed undirected graph $H$, also known as
the counting $H$-colouring problem.  Dyer and Greenhill have shown
that $\numCSP(\{E\})$ is in polynomial time if $E$ is a complete
relation or defines a complete bipartite graph and is \numP{}-complete
otherwise \cite{DG2000:num-Hcol}.  The dichotomy for directed acyclic
graphs has been characterized by Dyer, Goldberg and Paterson
\cite{DGP2007:num-DAGcol} and, more recently, Cai and Chen have shown
a dichotomy for all directed graphs, even with
non-negative algebraic weights \cite{CC2010:num-dircol}.  In contrast
to the decision problem, it is not known whether a direct proof of the
dichotomy for general directed graphs would yield an alternative proof
of the dichotomy for arbitrary constraint languages.

Restricting the tree-width of instances has a dramatic
effect.  In the case of counting $H$-colourings, restricting the
instance to be a graph of tree-width at most $k$ makes the problem
solvable in linear time for any graph $H$, a result due to
D{\'\i}az, Serna and Thilikos \cite{DST2002:H-col-TW}.  This result
follows immediately from Courcelle's theorem, which says that, if a
decision problem is definable in monadic second-order logic (which
$H$-colouring is, for any fixed $H$), then both it and the
corresponding counting problem are computable in linear time
\cite{Cou1990:MSO-tw,CMR2001:MSO-FPT}.  However, invocations of
Courcelle's theorem hide enormous constants in the notation
$\mathcal{O}(n)$ (in this case, a tower of twos of height $|H|$),
while the work of D{\'\i}az et al.\@ not only yields practical
constants but can also be applied to classes of instances where the
tree-width is allowed to grow logarithmically with the order of the
graph, rather than being constant.


\subsection{Approximate counting}

Since $\numCSP(\Gamma)$ is very often \numP{}-complete, approximation algorithms play an important role.  The key concept is that of a \emph{fully polynomial randomized approximation scheme} (FPRAS).  This is a randomized algorithm for computing some function $f(x)$, taking as its input $x$ and a constant $\epsilon > 0$, and computing a value $Y$ such that $e^{-\epsilon} \leq Y/f(x) \leq e^\epsilon$ with probability at least $\tfrac{3}{4}$, in time polynomial in both $|x|$ and ${\epsilon}^{-1}$. (See Section~\ref{sec:Prelim:Approx} for details.)

Dyer, Goldberg and Jerrum have classified the complexity of
approximately computing $\numCSP(\Gamma)$ for Boolean constraint
languages \cite{DGJ2010:Bool-approx}.  When all relations
in $\Gamma$ are affine, $\numCSP(\Gamma)$ can be computed exactly in
polynomial time by the result of Creignou and Hermann discussed above
\cite{CH1996:Bool-numCSP}.  Otherwise, if every relation in $\Gamma$
can be defined by a conjunction of Boolean implications and pins
(i.e., assertions of the form $v=0$ or $v=1$), then
$\numCSP(\Gamma)$ is as hard to approximate as the problem \numBIS{}
of counting independent sets in a bipartite graph; otherwise,
$\numCSP(\Gamma)$ is as hard to approximate as the problem \numSAT{}
of counting the satisfying truth assignments of a Boolean
formula.  Dyer, Goldberg, Greenhill and Jerrum have shown that the
latter problem is complete for \numP{} under appropriate
approximation-preserving reductions (see
Section~\ref{sec:Prelim:Approx}) and has no FPRAS unless $\NPtime =
\RPtime$ \cite{DGGJ2004:Approx}, which is thought to be unlikely.  The
complexity of \numBIS{} is currently open: there is no known FPRAS but
it is not known to be \numP{}-complete, either. \numBIS{} is known to
be complete with respect to approximation-preserving reductions in a
logically-defined subclass of \numP{} \cite{DGGJ2004:Approx}.


\subsection{Our result}

In this paper we consider the complexity of approximately solving
Boolean $\numCSP$ problems when instances have bounded
degree. Following Dalmau and Ford~\cite{DF2003:bdeg-gensat} and
Feder~\cite{Fed2001:Fanout} we consider the case in which
$\Rzero=\{0\}$ and $\Rone=\{1\}$ are available.  We show
that any Boolean relation that is not definable as a conjunction of
ORs or NANDs can be used in low-degree instances to assert equalities
between variables.  Thus, we can side-step degree restrictions by
replacing high-degree variables with distinct variables that are
constrained to be equal, reducing to Dyer, Goldberg and Jerrum's
trichotomy for Boolean
$\numCSP$ without degree restrictions~\cite{DGJ2010:Bool-approx}.

Our main result, Theorem~\ref{thmcor:main}, is a trichotomy for the
case in which instances have maximum degree~$d$ for any $d\geq
6$. If every relation in~$\Gamma$ is affine then $\numCSP_d(\Gamma
\cup \{\Rzero,\Rone\})$ is solvable in polynomial time. Otherwise, if
every relation in $\Gamma$ can be defined as a conjunction
of $\Rzero$, $\Rone$ and binary implications, then $\numCSP_d(\Gamma
\cup \{\Rzero,\Rone\})$ is equivalent in approximation complexity to
$\numBIS{}$. Otherwise, it has no FPRAS unless
$\NPtime=\RPtime$. Theorem~\ref{theorem:complexity} gives a partial
classification of the complexity when $d<6$. In the new cases that
arise here, the complexity is given in terms of $\numwHISd$, the
complexity of counting independent sets in hypergraphs of degree at
most~$d$ with hyper-edges of size at most $w$. The complexity of this
problem is not fully understood. We explain what is known about it in
Section~\ref{sec:Complexity}.


\subsection{Organization}

The remainder of the paper is organized as follows.  In
Section~\ref{sec:Prelim}, we define the basic notation, relational
operations and hypergraph properties that we use, and formally define
bounded-degree CSPs.  In Section~\ref{sec:Relations}, we introduce the
classes of relations that we will use throughout the paper and give
some of their basic properties.  A key tool in this type of
work~\cite{CLXxxxx:Complex-numCSP, Fed2001:Fanout} is characterizing
the ability of certain relations or sets of relations to assert
equalities between variables: we show when this can be done in
Section~\ref{sec:SimEq}.  The last piece of preparatory work is to
show that every Boolean relation that cannot simulate equality in this
way is definable by a conjunction of pins and either ORs or
NANDs, which is done in Section~\ref{sec:Trichotomy}.  Our
classification of the approximation complexity of bounded-degree
Boolean counting CSPs follows, in Section~\ref{sec:Complexity}.


\section{Preliminaries}
\label{sec:Prelim}


\subsection{Basic notation}
\label{sec:basicnotation}

We write $\abar$ for the tuple $\tuple{a_1, \dots, a_r}$, which we
often shorten to $a_1\dots a_r$.  We write $a^r$ for the
$r$-tuple $a\dots a$ and $\abar\bbar$ for the tuple formed from
the elements of $\abar$ followed by those of $\bbar$.

The \emph{bit-wise complement} of a relation $R\subseteq \Bool^r$
is the relation
\begin{equation*}
    \bwcomp{R} = \{\tuple{a_1\oplus 1, \dots, a_r\oplus 1}
                                                   \mid \abar\in R\}\,,
\end{equation*}
where $\oplus$ denotes addition modulo~2.

We say that a relation $R$ is \emph{ppp-definable}\footnote{This
should not be confused with the concept of primitive positive
definability (pp-definability) which appears in algebraic
treatments of CSP and \numCSP{}, for example in the work of
Bulatov \cite{Bul2008:Dichotomy}.} in a relation $R'$ and write
$R\pppleq R'$ if $R$ can be obtained from $R'$ by some sequence of the
following operations:
\begin{itemize}
\itemspacing
\item permutation of columns;
\item pinning (taking sub-relations of the form $R_{i\mapsto c} =
    \{\abar\in R \mid a_i = c\}$ for some $i$ and some $c\in\Bool$); and
\item projection (``deleting the $i$th column'' to give
    $\{a_1\dots a_{i-1} a_{i+1}\dots a_r \mid a_1\dots a_r\in R\}$).
\end{itemize}
The three p's in ``ppp-definable'' refer to the initial letters of the
words permutation, pinning and projection. Allowing permutation of
columns is just a notational convenience: it clearly adds no
expressive power.

It is easy to see that $\pppleq$ is a partial order on Boolean
relations and that, if $R\pppleq R'\!$, then $R$ can be obtained from
$R'$ by first permuting the columns, then making some pins and then
projecting.

We write $\Rzero = \{0\}$, $\Rone = \{1\}$, $\Req = \{00, 11\}$, $\Rneq = \{01, 10\}$, $\Ror = \{01, 10,
11\}$, $\Rnand = \{00, 01, 10\}$, $\Rimp = \{00, 01, 11\}$ and $\Rpmi=\{00, 10, 11\}$.  For
$k\geq 2$, we write $\Reqk = \{0^k\!, 1^k\}$, $\Rork = \Bool^k
\setminus \{0^k\}$ and $\Rnandk = \Bool^k \setminus \{1^k\}$ (i.e.,
$k$-ary equality, \OR{} and \NAND{}, respectively).

We write $\proj{i}R$ for the projection of $R$ onto its $i$th column
and $\proj{i,j}R$ for the projection onto columns $i$ and~$j$.


\subsection{Boolean constraint satisfaction problems}

A \emph{constraint language} is a set $\Gamma = \{R_1, \dots, R_m\}$
of named Boolean relations.  Given a set $V$ of variables, a
\emph{constraint} over $\Gamma$ is an expression $R(\vbar)$ where
$R\in\Gamma$ has arity $r$ and $\vbar\in V^r\!$.  Note that, if $v$ and
$v'$ are variables, neither $v=v'$ nor $v\neq v'$ is a constraint,
though of course $\Req(v,v')$ is a constraint if $\Req \in \Gamma$ and
similarly for $\Rneq$. The \emph{scope} of a constraint $R(\vbar)$ is
the tuple $\vbar$.  Note that the variables in the scope of a
constraint need not all be distinct.

An \emph{instance} of the constraint satisfaction problem (CSP) over
$\Gamma$ is a set $V$ of variables and a set $C$ of constraints over
$\Gamma$ in the variables in $V$.

An \emph{assignment} to a set $V$ of variables is a function
$\sigma\colon V\to \Bool$ and it \emph{satisfies} an
instance $(V, C)$ if $\tuple{\sigma(v_1), \dots, \sigma(v_r)}\in R$
for every constraint of the form $R(v_1, \dots, v_r)$. Given an
instance $I$ of some CSP, we write $Z(I)$ for the number of satisfying
assignments.

We are interested in the counting CSP problem $\numCSP(\Gamma)$
(parameterized by $\Gamma$), defined as:
\begin{description}
\item[Input:] an instance $I=(V, C)$ of CSP over $\Gamma$.
\item[Output:] $Z(I)$.
\end{description}

The \emph{degree} of an instance is the greatest number of times any
variable appears among its constraints.  Note that the variable $v$
appears twice in the constraint $R(v,v)$.  Our specific interest in
this paper is in classifying the complexity of bounded-degree counting CSPs.
For a constraint language $\Gamma$ and a positive integer $d$, define
$\numCSPd(\Gamma)$ to be the restriction of $\numCSP(\Gamma)$ to
instances of degree at most $d$.
We can deal with instances of degree~1 immediately.

\begin{theorem}
\label{thrm:degree-1}
    For any $\Gamma\!$, $\numCSPd[1](\Gamma)\in \FPtime$.
\end{theorem}
\begin{proof}
    Because each variable appears at most once, the constraints are
    independent.  Each constraint $R(v_1, \dots, v_r)$ can be
    satisfied in $|R|$ ways and any variable that does not appear in a
    constraint can take either the value 0 or~1.  The total number of
    assignments is the product of the number of ways each constraint
    can be satisfied, multiplied by $2^k\!$, where $k$ is the number of
    unconstrained variables.
\end{proof}

A key technique in proving hardness results for \numCSP{} and related
problems is \emph{pinning} \cite{CH1996:Bool-numCSP, DG2000:num-Hcol,
  Fed2001:Fanout, DF2003:bdeg-gensat, DGJ2009:WBool,
  DGJ2010:Bool-approx}.  We write $\Rzero=\{0\}$ and $\Rone=\{1\}$ for
the two unary relations that contain only zero and one,
respectively.  We refer to constraints in $\Rzero$ and $\Rone$ as
\emph{pins} and we say that the single variable in the scope of a pin
is \emph{pinned}.  To make notation easier, we will sometimes write
constraints using constants instead of explicit pins.  That is, we
will write constraints of the form $R(x_1, \dots, x_r)$ where each
$x_i$ is either a variable from $V$ or a constant 0 or~1 (again, the
$x_i$ need not be distinct).  Such a constraint can always be
rewritten as a set of ``proper'' constraints by replacing each instance
of a constant 0 or~1 with a fresh variable $v$ and introducing the
appropriate constraint $\Rzero(v)$ or~$\Rone(v)$.  Note that every
variable introduced in this way appears exactly twice in the resulting
instance so if the degree of the CSP instance is at least two, the
transformation does not increase the instance's degree. We let
$\GammaPin$ denote the constraint language $\{\Rzero, \Rone\}$.

When there are no degree bounds, adding pinning does not affect
complexity results for either the exact or approximate version of
\numCSP{}.  In the exact case, the addition of pinning does not affect
the structural properties that determine the complexity of
$\numCSP(\Gamma)$ \cite{DR2011:dichotomy} whereas, for approximation
on the Boolean domain, there are reductions of the appropriate kind
from $\numCSP(\Gamma \cup \GammaPin)$ to $\numCSP(\Gamma)$
\cite{DGJ2010:Bool-approx, DGJ2009:WBool}.  However, these reductions
increase the degree of variables so are not applicable in our setting.
In order to make progress,
we follow earlier work on degree-bounded CSP \cite{Fed2001:Fanout,
DF2003:bdeg-gensat} and assume that pinning is available in
constraint languages.  This plays a significant role in
Section~\ref{sec:SimEq}.


\subsection{Hypergraphs}

A \emph{hypergraph} $H=(V,E)$ consists of a set $V=V(H)$ of vertices
and a set $E = E(H)\subset \powerset{V}$ of non-empty
\emph{hyper-edges}. The \emph{degree} of a vertex $v\in V(H)$ is the
number $d(v)$ of hyper-edges it participates in: $d(v) = |\{e\in
E(H)\mid v\in e\}|$. The degree of a hypergraph is the maximum degree
of its vertices. If $w = \max \{|e| \mid e\in E(H)\}$, we say that $H$
has \emph{width} $w$.

An \emph{independent set} in a hypergraph $H$ is a set $S\subseteq
V(H)$ such that $e\nsubseteq S$ for every $e\in E(H)$.  Notice
that we may have more than one vertex of a hyper-edge in an
independent set, so long as at least one vertex of each hyper-edge is
omitted.

We write \numwHIS{} for the following problem:
\begin{description}
    \setlength{\itemsep}{0pt}
    \item[Input:]  a width-$w$ hypergraph $H$
    \item[Output:] the number of independent sets in $H$
\end{description}
and \numwHISd{} for the following problem:
\begin{description}
    \setlength{\itemsep}{0pt}
    \item[Input:]  a width-$w$ hypergraph $H$ of degree at most $d$
    \item[Output:] the number of independent sets in $H$.
\end{description}


\subsection{Approximation complexity}
\label{sec:Prelim:Approx}

A \emph{randomized approximation scheme} (RAS) for a function $f\colon\Sigma^*\to\mathbb{N}$ is a probabilistic Turing machine that takes as input a pair $(x,\epsilon)\in \Sigma^*\times (0,1)$, and produces, on an output tape, an integer random variable~$Y$ satisfying the condition $\Pr(e^{-\epsilon} \leq Y/f(x) \leq e^\epsilon)\geq \frac{3}{4}$.%
       \footnote{The choice of the value $\frac{3}{4}$ is inconsequential: the same class of problems has an FPRAS if we choose any probability $\frac{1}{2}<p<1$ \cite{JVV1986:Randgen}.}
A \emph{fully polynomial randomized approximation scheme (FPRAS)} is a
RAS that runs in time polynomial in both $|x|$ and $\epsilon^{-1}\!$.

To compare the complexity of approximate counting problems, we use the
AP-reductions of \cite{DGGJ2004:Approx}. Suppose that $f$ and $g$ are
functions from some input domain $\Sigma^*$ to the natural numbers and
we wish to compare the complexity of approximately computing them.  An
\emph{approximation-preserving} reduction from~$f$ to~$g$ is a
probabilistic oracle Turing machine $M$ whose input is a pair
$(x,\epsilon)\in \Sigma^*\times (0,1)$, and which satisfies the following
three conditions: (i) every oracle call made by $M$ is of the form
$(w,\delta)$ where $w\in \Sigma^*$ is an instance of~$g$ and
$0<\delta<1$ is an error bound satisfying $\delta^{-1} \leq
\mathrm{poly}(|x|,\epsilon^{-1})$; (ii) $M$ is a randomized
approximation scheme for $f$ whenever the oracle is a randomized
approximation scheme for $g$; and (iii) the running time of $M$ is
polynomial in $|x|$ and $\epsilon^{-1}$.
 
If there is an approximation-preserving reduction from $f$ to $g$, we write $f\APred g$ and say that $f$ is \emph{AP-reducible} to $g$. If $g$ has an FPRAS then so does $f$. If $f\APred g$ and $g\APred f$ then we say that $f$ and $g$ are \emph{AP-interreducible} and write $f\APequiv g$.

AP-reductions are well-suited to approximate counting problems.  The
class of problems admitting an FPRAS is closed under these reductions
and a Ladner-like hierarchy of AP-interreducible approximation
problems has been shown to exist~\cite{Bor2010:Approx-hierarchy}.
Further, the intuition that the counting version of an
\NPtime{}-complete problem should be \numP{}-complete is a theorem if
\numP{}-completeness is defined with respect to
AP-reductions~\cite{DGGJ2004:Approx} but is not known to hold for
other candidate classes of reduction, such as Simon's parsimonious
reductions~\cite{Sim1977:One-vs-many} and polynomial-time Turing
reductions.


\section{Classes of relations}
\label{sec:Relations}

A relation $R\subseteq \Bool^r$ is \emph{affine} if it is the set of solutions to some system of linear equations over \GFtwo{}.  That is, there is a set $\Sigma$ of equations in variables $x_1, \dots, x_r$ where each equation has the form $\bigoplus_{i\in I}x_i = c$, where $\oplus$ denotes addition modulo~2, $I\subseteq [1,r]$ and $c\in \Bool$, and we have $\abar\in R$ if, and only if, the assignment $x_1\mapsto a_1, \dots, x_r\mapsto a_r$ satisfies every equation in $\Sigma$.  Note that the empty relation is defined by the equation $0=1$ (or, more formally, $\bigoplus_{i \in \emptyset} = 1$) and the complete relation $\Bool^r$ is defined by the empty set of equations. If a variable $x_i$ occurs in an equation of the form $x_i = c$, we say that it is \emph{pinned to $c$}.


\subsection{\ORconj{}, \NANDconj{}, \IMconj{} and normalized formulae}

Let \ORconj{} be the set of Boolean relations that are defined by
conjunctions of pins and \OR{}s of any arity and let \NANDconj{} be the
set of Boolean relations definable by conjunctions of pins and
\NAND{}s (i.e., negated conjunctions) of any arity. For example, the
8-ary relation defined by the formula
    \begin{equation*}
        (x_1 = 0) \ \wedge \ (x_2 = 1) \ \wedge \ \OR(x_3,x_4,x_5,x_6) \
            \wedge \ \OR(x_5,x_8)
    \end{equation*}
is in \ORconj{}.  (Note, also, that it does not constrain the variable $x_7$.) We say that one of the defining formulae of these relations is \emph{normalized} if
\begin{itemize}
\itemspacing
\item no pinned variable appears in any \OR{} or \NAND{},
\item the arguments of each individual \OR{} and \NAND{} are distinct,
\item every \OR{} or \NAND{} has at least two arguments and
\item no \OR{} or \NAND{}'s arguments are a subset of any
    other's.
\end{itemize}
Note that the formula in the example above is normalized.

\begin{lemma}
\label{lemma:conj-norm}
    Every \ORconj{} (respectively, \NANDconj{}) relation is defined by
    a unique normalized formula.
\end{lemma}
\begin{proof}
    We show the result for \ORconj{} relations; the case for
    \NANDconj{} is similar.

    Let $R$ be an \ORconj{} relation defined by the formula $\phi$.
    The second and subsequent occurrences of any variable within a
    single clause can be deleted.
    Any clause that contains a variable pinned to one can be deleted;
    any variable that is pinned to zero can be deleted from any clause
    in which it appears.  The disjunction $\OR(x)$ is equivalent to
    pinning $x$ to one.  If $\phi$ contains a clause that is a subset
    of another, any assignment that satisfies the smaller clause
    necessarily satisfies the latter, which can, therefore, be
    deleted.  This establishes that every \ORconj{} relation is
    defined by at least one normalized formula.

    To prove uniqueness, suppose that the relation $R\subseteq \Bool^r$ is defined
    by the normalized formulae $\phi$ and $\psi$.  The two formulae must
    obviously pin the same variables and we may assume that none are
    pinned.  Consider any clause in $\phi$, which we may assume,
    without loss of generality, to be $\OR(x_1, \dots, x_k)$.  Since
    no clause of $\phi$ is a subset of $\{x_1, \dots, x_k\}$, every
    other clause must include at least one variable from $x_{k+1},
    \dots, x_r$ and, therefore, $0^{k-1}1^{r-k+1}$ satisfies $\phi$
    and $0^k1^{r-k}$ does not.

    Now, suppose that this clause does not appear in $\psi$.  There
    are two cases.  If $\psi$ contains a clause whose variables are a
    subset of $\{x_1, \dots, x_k\}$, which we may assume, without loss
    of generality, to be $\OR(x_1, \dots, x_\ell)$ for some $\ell<k$,
    then $\psi$ is not satisfied by $0^{k-1}1^{r-k+1}\!$.  Otherwise,
    every clause of $\psi$ contains at least one variable from
    $x_{k+1}, \dots, x_r$, so $0^k1^{r-k}$ satisfies $\psi$.
    In either case, $\phi$ and $\psi$ define different relations.  It
    follows that every clause that appears in $\phi$ must also appear
    in $\psi$.  By symmetry, every clause that appears in $\psi$ must
    appear in $\phi$ so the two formulae are identical.
\end{proof}

Given the uniqueness of defining normalized formulae, we define the
\emph{width} of an \ORconj{} or \NANDconj{} relation $R$ to be
$\width{R}$, the greatest number of arguments to any of the \OR{}s or
\NAND{}s in the normalized formula that defines it.  Note that, from
the definition of normalized formulae, there are no relations of
width~1.  However, a conjunction of pins can be seen as an \ORconj{}
formula with no \OR{}s, i.e., of width~0: such a formula defines the
complete relation, possibly padded with some constant columns.  A
conjunction of pins is also a \NANDconj{} formula with no \NAND{}s so
we will usually just refer to these relations as ``relations of
width~0.''  We define the width of an \ORconj{} or \NANDconj{}
constraint language to be the greatest width of the relations within
it.

We define \IMconj{} to be the class of relations defined by
conjunctions of pins and (binary) implications --- this class is called
$\text{IM}_2$ in \cite{DGJ2010:Bool-approx}.  We say that a conjunction
of pins and implications is \emph{normalized} if no pinned variable
appears in an implication and the arguments of every implication are distinct.

\begin{lemma}
\label{lemma:IMconj-norm}
    Every relation in \IMconj{} is defined by a normalized formula.
\end{lemma}
\begin{proof}
    Let $R\in\IMconj{}$ be defined by the formula $\phi$.  Any
    implication $x\imp x$ can be deleted as it does not constrain
    the value of $x$.  If the variable $y$ is pinned to zero then any
    implication $y\imp z$ can be deleted and any implication
    $z\imp y$ can be replaced by pinning $z$ to zero.  If $y$ is
    pinned to one, $y\imp z$ can be replaced by pinning $z$ to one
    and $z\imp y$ can be deleted.  Iterating, we can remove all
    implications involving pinned variables.
\end{proof}

Note that, in contrast to normalized \ORconj{} and \NANDconj{}
formulae, normalized \IMconj{} formulae are not necessarily unique.
For example, the following three normalized formulae all define the
same relation:
\begin{gather*}
    x\imp y \ \wedge \ y\imp z \ \wedge \ z\imp x                     \\
    x\imp z \ \wedge \ z\imp y \ \wedge \ y\imp x                     \\
    x\imp y \ \wedge \ y\imp x \ \wedge \ x\imp z \ \wedge \ z\imp x\,.
\end{gather*}


\subsection{ppp-defining Boolean connectives}

\begin{lemma}
\label{lemma:IMconj-implies}
    If $R\in\IMconj$ is not affine, then $\Rimp\pppleq R$.
\end{lemma}
\begin{proof}
    Let $R\in\IMconj{}$ be defined by the normalized formula $\phi$.
    If there are variables $x_1, \dots, x_r$ such that $\phi$ contains
    the implications $x_1\imp x_2$, \dots, $x_{r-1}\imp x_r$ and
    $x_r\imp x_1$ then, in any satisfying assignment for $\phi$, the
    variables $x_1, \dots, x_r$ must take the same value.  Hence, we may
    assume that, if $\phi$ contains such a cycle of
    implications, it also contains $x_i\imp x_j$ for every distinct
    pair $x_i, x_j\in\{x_1, \dots, x_r\}$.

    There are two cases.  First, if $\phi$ is symmetric (in the sense
    that, for every implication $x\imp y$ in $\phi$, the formula also
    contains $y\imp x$) then $\phi$ is equivalent to a conjunction of
    pins and equalities between variables, so $R$ is affine.
    Otherwise, there must be at least one pair of variables such that
    $x\imp y$ is a conjunct of $\phi$ but $y\imp x$ is not.  We
    ppp-define implication by pinning to zero every unpinned variable
    $v_1$ such that there is a chain of implications
    $v_1\imp v_2$, \dots, $v_{r-1}\imp v_r$, $v_r\imp x$ and pinning
    to one every other unpinned variable apart from $x$ and $y$.
    Finally, project out the $r-2$ constant columns.
\end{proof}

\begin{lemma}
\label{lemma:ORconj-OR}
    If $R\in\ORconj$ has width $w$, then $\Rork[2], \dots, \Rork[w] \pppleq R$. Similarly, if $R\in\NANDconj$ has width $w$, then $\Rnandk[2], \dots, \Rnandk[w] \pppleq R$.
\end{lemma}
\begin{proof}
    Let $R\in\ORconj{}$ have arity $r$ and width $w$.  Let $R$ be
    defined by the normalized formula $\phi$ which, without loss of
    generality, we may assume to contain the clause $\OR(x_1, \dots,
    x_w)$.  Since $\phi$ is normalized, every other clause must
    contain at least one variable from $x_{w+1}, \dots, x_r$.  For any
    $k$ with $2\leq k\leq w$, we can ppp-define $\Rork$ by pinning
    $x_{k+1}, \dots, x_w$ to zero and pinning $x_{w+1}, \dots, x_r$ to
    one. The proof for $R\in\NANDconj{}$ is similar.
\end{proof}


\subsection{Characterizations}

The following proposition establishes a duality between \ORconj{} and
\NANDconjLinebreak{} relations. Whenever we say that $R$ is \ORconj{}
or \NANDconj{}, it is equivalent to say that $R$ or $\bwcomp{R}$ is
\ORconj, where $\bwcomp{R}$ is the bit-wise complement of~$R$, as defined in Section~\ref{sec:basicnotation}.
Of course, it is also equivalent to say  that $R$ or $\bwcomp{R}$ is \NANDconj{}. 

\begin{proposition}
\label{prop:OR-NAND}
    A relation $R\subseteq \Bool^r$ is in \ORconj{} if, and only if,
    $\bwcomp{R}\in \NANDconj{}$.
\end{proposition}
\begin{proof}
    Suppose $R$ is defined by the normalized formula
    \begin{equation*}
        P \ \ \wedge \bigwedge_{1\leq j\leq m} \bigvee_{i\in I_j} x_i\,,
    \end{equation*}
    where $P$ is a conjunction of pins and $I_1, \dots, I_m\subseteq
    [1,r]$.  Then $\bwcomp{R}$ is defined by the formula
    \begin{equation*}
        P' \ \ \wedge \bigwedge_{1\leq j\leq m}
                                         \bigvee_{i\in I_j} \neg x_i\,,
    \end{equation*}
    where $P'$ is the conjunction of pins with the opposite values to
    those in $P$.  This formula is equivalent to
    \begin{equation*}
        P' \ \ \wedge \bigwedge_{1\leq j\leq m}
                                       \neg \bigwedge_{i\in I_j} x_i\,,
    \end{equation*}
     which is a \NANDconj{} formula, as required.  The argument is
     reversible.
\end{proof}

Given tuples $\abar, \bbar\in \Bool^r\!$, we write $\abar\leq \bbar$
if $a_i\leq b_i$ for all $i\in [1,r]$.  If $\abar\leq \bbar$ and
$\abar \neq \bbar$, we write $\abar < \bbar$.
We say that a relation $R\subseteq \Bool^r$ is \emph{monotone} if,
whenever $\abar\in R$ and $\abar\leq \bbar$, then $\bbar\in R$.  We
say that $R$ is \emph{antitone} if, whenever $\abar\in R$ and
$\bbar\leq \abar$, then $\bbar\in R$.  That is, changing zeroes to
ones in a tuple in a monotone relation gives another tuple in the
relation; similarly,
antitone relations are preserved by changing ones to zeroes.  It is
easy to see that $R$ is monotone if, and only if, $\bwcomp{R}$ is
antitone.  We say that a relation is \emph{pseudo-monotone}
(respectively, \emph{pseudo-antitone}) if its restriction to
non-constant columns is monotone (respectively, antitone).  The
following is a simple consequence of results in
\cite[Section~7.1.1]{KnuXXXX:TAOCPv4A}.

\begin{proposition}
\label{prop:OR-monotone}
    A relation $R\subseteq \Bool^r$ is in \ORconj{} (respectively,
    \NANDconjLinebreak) if, and only if, it is pseudo-monotone (respectively,
    pseudo-antitone).
\end{proposition}


\section{Simulating equality}
\label{sec:SimEq}

An important ingredient in bounded-degree dichotomy theorems
\cite{Fed2001:Fanout,CLXxxxx:Complex-numCSP} is showing how to express
equality using constraints from a constraint language that does not
necessarily include the equality relation. In this section, we give
the definitions that we need and some results about when equality can
be expressed in our setting.

Recall that, for all integers $k \geq 2$, $\Reqk$ is the $k$-ary
equality relation $\{0^k\!, 1^k\}$.  We say that a constraint language
$\Gamma$ \emph{simulates} $\Reqk$ if, for some $\ell \geq k$ there is
an integer $m \geq 1$ and a $(\Gamma \cup \GammaPin)$-CSP instance $I$
with variables $x_1, \dots, x_\ell$ and such that $I$ has exactly $m$
satisfying assignments $\sigma$ with $\sigma(x_1) = \dots =
\sigma(x_k) = 0$, exactly $m$ with $\sigma(x_1) = \dots = \sigma(x_k)
= 1$ and no other satisfying assignments.  If, further, the degree of
$I$ is $d$ and the degree of each variable $x_1, \dots, x_k$ is at
most $d-1$, we say that $\Gamma$ \emph{simulates $\Reqk$ with $d$
  variable repetitions} or, for brevity, that $\Gamma$
\emph{$d$-simulates} $\Reqk$.  We say that $\Gamma$
\emph{$d$-simulates equality} if it $d$-simulates $\Reqk$ for all
$k\geq 2$. If only one relation $R$ is involved in the simulation, we
drop the curly brackets and say that $R$, rather than $\{R\}$,
$d$-simulates equality.

The point of this slightly strange definition is that, if $\Gamma$
$d$-simulates equality, we can express the constraint $y_1 = \dots =
y_k$ in $\Gamma \cup \GammaPin$ and then use each $y_i$ in one further
constraint, while still having an instance of degree $d$.  The
variables $x_{k+1}, \dots, x_\ell$ in the definition function as
auxiliary variables and do not appear in any other constraint.  This
means that, if the variable $y$ occurs $k>d$ times in some instance,
we can replace the successive occurrences with distinct variables
$y_1, \dots, y_k$ that are constrained to be equal, giving an
equivalent instance of degree at most $d$.

Concepts similar to simulation have been used before, such as
``perfect implementation''~\cite{CKS2001:Bool-CSP} and
``implementation''~\cite{DGJ2010:Bool-approx}.  The difference is that
our setting demands degree bounds on the constraints used in
simulation and, for counting, we need to preserve the number of
satisfying assignments (at least, up to some constant multiple) not
just the existence of satisfying assignments.

\begin{proposition}
\label{prop:bound-unbound}
    If $\Gamma$ $d$-simulates equality, then $\numCSP(\Gamma) \APredto
    \numCSPd(\Gamma \cup \GammaPin)$.
\end{proposition}
\begin{proof}
    Let $I$ be an instance of $\numCSP(\Gamma)$.  We produce a new CSP
    instance $I'$ over the constraint language $\Gamma$ augmented with
    $\Reqk[i]$ constraints for certain values of $i$ as follows.  For
    each variable $x$ that appears $k>d$ times in $I$, replace the
    occurrences with new variables $x_1, \dots, x_k$ and add the
    constraint $\Reqk(x_1, \dots, x_k)$.  Clearly, $Z(I') = Z(I)$.

    Note that every variable in $I'$ either occurs exactly once in an
    equality constraint (one of the form $\Reqk[i](\xbar)$) and
    exactly once in a $\Gamma$-constraint or occurs in no equality
    constraints and at most $d$ times in $\Gamma$-constraints.  Since
    $\Gamma$ $d$-simulates equality, we can replace the
    equality constraints with $(\Gamma \cup \GammaPin)$-constraints,
    using fresh auxiliary variables for each equality, to give an
    instance $I''$ of $\numCSP(\Gamma \cup \GammaPin)$ with degree
    $d$.  There is some constant $m$, depending only on the number and
    arities of the equality constraints in $I'\!$, such that $Z(I'') =
    mZ(I')$.  Since $m$ can be computed in polynomial time, we have an
    AP-reduction.
\end{proof}

\begin{lemma}
\label{lemma:proj-eq}
    Let $R\subseteq\Bool^r\!$.  If $\Req\pppleq R$, $\Rneq\pppleq R$ or
    $\Rimp\pppleq R$, then $R$ 3-simulates equality.
\end{lemma}

Note that, if $\Rpmi\pppleq R$ then $\Rimp\pppleq R$, also.

\begin{proof}[Proof of Lemma~\ref{lemma:proj-eq}]
    For each $k\geq 2$, we show how to 3-simulate $\Reqk$.  We may
    assume without loss of generality that the ppp~definition of
    $\Req$, $\Rneq$ or $\Rimp$ from $R$ involves applying the identity
    permutation to the columns, pinning columns 3 to $3+p-1$ inclusive
    to zero, pinning columns $3+p$ to $3+p+q-1$ inclusive to one (that
    is, pinning $p\geq 0$ columns to zero and $q\geq 0$ to one)
    and then projecting away all but the first two columns.

    Suppose first that $\Req\pppleq R$ or $\Rimp\pppleq R$.  $R$
    must contain $\alpha\geq 1$ tuples that begin $000^p1^q$,
    $\beta\geq 0$ that begin $010^p1^q$ and $\gamma\geq 1$ that begin
    $110^p1^q$, and we have $\beta=0$ unless we are ppp-defining $\Rimp$.

    We consider, first, the case where $\alpha=\gamma$, and show that
    we can 3-simulate $\Reqk$, expressing the constraint $\Reqk(x_1,
    \dots, x_k)$ with the constraints
    \begin{equation*}
        R(x_1 x_2 0^p 1^q *), \ R(x_2 x_3 0^p 1^q *), \dots,
	    \ R(x_{k-1} x_k 0^p 1^q *), \ R(x_k x_1 0^p 1^q *) \,,
    \end{equation*}
    where $*$ denotes a fresh $(r-2-p-q)$-tuple of variables in each
    constraint.  This set of constraints is equivalent to either $x_1
    = \dots = x_k = x_1$ or $x_1\imp \cdots \imp x_k \imp x_1$ so, in either
    case, constrains the variables $x_1, \dots, x_k$ to have the same
    value, as required.  Every variable appears at most twice and
    there are $\alpha^k$ solutions to these constraints that put
    $x_1=\dots=x_k=0$, the same number with $x_1=\dots=x_k=1$ and no
    other solutions.  Therefore, $R$ 3-simulates $\Reqk$.

    We now show, by induction on $r$, the arity of $R$, that we can
    3-simulate $\Reqk$ even if $\alpha$ is not necessarily equal to
    $\gamma$.  For the base case, $r=2$, we have $\alpha=\gamma=1$ and
    we are done.  For the inductive step, let $r>2$ and assume,
    without loss of generality that $\alpha>\gamma$ (we are already
    done if $\alpha=\gamma$ and the case $\alpha<\gamma$ is
    symmetric).  In particular, we have $\alpha\geq 2$, so there are
    distinct tuples $000^p1^q\abar$ and $000^p1^q\bbar$ in $R$.  $R$
    also contains a tuple $110^p1^q\cbar$.  Choose $j$ such that
    $a_j\neq b_j$.  Pinning the $(2+p+q+j)$th column of $R$ to $c_j$
    and projecting out the resulting constant column gives a relation
    of arity $r-1$ that still contains at least one tuple beginning
    $000^p1^q$ and at least one beginning $110^p1^q$: by the
    inductive hypothesis, this relation 3-simulates $\Reqk$.

    Finally, we consider the case that $\Rneq\pppleq R$.  $R$ contains
    $\alpha\geq 1$ tuples beginning $010^p1^q$ and $\beta\geq 1$
    beginning $100^p1^q$ and no other tuples.  We express the
    constraint $\Reqk(x_1,
    \dots, x_k)$ by introducing fresh variables $y_1, \dots, y_k$ and
    using the constraints
    \begin{gather*}
                R(x_1 y_1 0^p1^q*), R(y_1 x_2 0^p1^q*),             \\
                R(x_2 y_2 0^p1^q*), R(y_2 x_3 0^p1^q*),             \\
	                        \vdots                               \\
        R(x_{k-1} y_{k-1} 0^p1^q*), R(y_{k-1} x_k 0^p1^q*),         \\
                R(x_k y_k 0^p1^q*), R(y_k x_1 0^p1^q*)\,,
    \end{gather*}
    where $*$ denotes a fresh $(r-2-p-q)$-tuple of variables in each
    constraint, as before.  These constraints have $\alpha^k \beta^k$
    solutions with $x_1 = \dots = x_k = 0$ and $y_1 = \dots = y_k = 1$
    and $\beta^k \alpha^k$ solutions that assign 1 to all the $x$'s
    and 0 to all the $y$'s.  There are no other solutions and no
    variable is used more than twice.
\end{proof}

The following technical lemma and the definitions that support it are
used only to prove Lemma~\ref{lemma:valid-eq}. For $c\in \Bool$, an
$r$-ary relation is \emph{$c$-valid} if it contains the tuple $c^r\!$.  Given
a relation $R\subseteq \Bool^r\!$, a tuple $\abar\in R$ that contains
both zeroes and ones and a constant $c\in\Bool$, let $R_{\abar, c}$ be
the result of pinning the set of columns $\{i\mid a_i=c\}$ to $c$ and
then projecting out those columns.  Observe that $R_{\abar, c}$ is
always $(1-c)$-valid (because it contains the projection of $\abar$)
and is $c$-valid if $R$ is (because then it contains the projection of
$c^r$).

\begin{lemma}
\label{lemma:R-abar-c}
    Let $r\geq 3$ and let $\Reqk[r]\subsetneq R \subsetneq
    \Bool^r\!$. There are $\abar\in R$ and $c\in\Bool$ such that
    $R_{\abar,c}$ is not complete.
\end{lemma}
\begin{proof}
    Suppose there is a tuple $\abar\in R\setminus\{0^r\}$ such that
    changing some zero in $\abar$ to a one gives a tuple $\abar'\notin
    R$.  Then $R_{\abar, 1}$ does not contain the relevant projection
    of $\abar'$ and we are done.  Similarly, if there is a tuple
    $\bbar\in R\setminus \{1^r\}$ that leaves $R$ by changing some one
    to a zero, then $R_{\bbar, 0}$ is not complete.  If no such tuple
    exists, then either $R=\Bool^r$ or $R=\Reqk[r]$, contradicting our
    assumptions.
\end{proof}

\begin{lemma}
\label{lemma:valid-eq}
    Let $r\geq 2$ and let $R\subset \Bool^r$ be 0- and 1-valid but not
    complete. Then $R$ 3-simulates equality.
\end{lemma}
\begin{proof}
    We show by induction on $r$ that either $\Req$ or $\Rimp$ is
    ppp-definable in $R$, and the result follows by
    Lemma~\ref{lemma:proj-eq}.

    In the case $r=2$, $R$ is either $\Req$, $\Rimp$ or $\Rpmi$.
    For $r\geq 3$, if
    $R=\Reqk[r]$ then $\proj{1,2} R = \Req$.  Otherwise, by Lemma~\ref{lemma:R-abar-c},
    there is some $\abar\in R$ and $c\in\Bool$ such that $R_{\abar,c}$
    is not complete.  Since $R_{\abar,c}\pppleq R$ and is 0- and
    1-valid, we are done by the inductive hypothesis.
\end{proof}

We will next show that, if binary OR is ppp-definable in $R$ and
binary NAND in $R'\!$, then the constraint language $\{R, R'\}$
3-simulates equality ($R$ and $R'$ need not be distinct).  To do this,
we will use the following sets of
constraints, $\xi_k$, for $k\geq 2$:
\begin{align*}
    \xi_k = &\ \{\Ror(x_i, y_i) \mid 1\leq i\leq k\} \\
            &\ \cup \{\Rnand(y_i, x_{i+1}) \mid 1\leq i< k\}
                \cup \{\Rnand(y_k, x_{1})\}\,.
\end{align*}

The key point about these constraints is that they show that the
language $\{\Ror, \Rnand\}$ 3-simulates equality.

\begin{lemma}
\label{lemma:xi-sat}
    An assignment $\sigma$ to $\{x_1, \dots, x_k, y_1, \dots, y_k\}$
    satisfies all constraints in $\xi_k$ if, and only if, $\sigma(x_1)
    = \dots = \sigma(x_k) \neq \sigma(y_1) = \dots = \sigma(y_k)$.
\end{lemma}
\begin{proof}
    It is easy to check that assignments of the given type satisfy
    $\xi_k$.  Conversely, suppose that $\sigma$ satisfies $\xi_k$.

    If $\sigma(x_1) = 0$, we have $\sigma(y_1) = 1$ because
    $\Ror(x_1,y_1)$ is satisfied and we must have $\sigma(x_2) = 0$
    because $\Rnand(y_1, x_2)$ is satisfied.  By a trivial induction,
    $\sigma(x_i) = 0$ and $\sigma(y_i) = 1$ for all $i$.

    Otherwise, $\sigma(x_1) = 1$.  If $\sigma(x_i) = 0$ for any $i>1$
    then, by the same argument as above, $\sigma(x_i) = 0$ for all
    $i\in [1,k]$, contradicting the assumption that $\sigma(x_1) = 1$.
    Therefore, $\sigma(x_i) = 1$ for all $i$.  To satisfy the
    constraints $\Rnand(y_i, x_{i+1})$, we must have $\sigma(y_i) = 0$
    for all $i$.
\end{proof}

We now show that, in fact, we do not need to have $\Ror$ and $\Rnand$
in our constraint language $\Gamma$: it suffices to be able to
ppp-define them from relations in $\Gamma\!$.

\begin{lemma}
\label{lemma:OR-NAND-eq}
    If $\Ror\pppleq R$ and $\Rnand\pppleq R'$ then $\{R,R'\}$ 3-simulates equality.
\end{lemma}
\begin{proof}
    Suppose first that $R$ and $R'$ are two distinct relations.
    We may assume, as in the proof of Lemma~\ref{lemma:proj-eq}, that
    the ppp definition of $\Ror$ from
    $R$ involves performing some permutation and projecting to the
    first two columns after pinning the next $p$ columns to zero and
    the $q$ columns after that to one.  We may suppose further that
    we cannot pin any more columns of $R$ and still
    ppp-define $\Ror$.  Without loss of generality, we may assume the
    permutation to be the identity.

    Under these assumptions, $R$ contains $\alpha\geq 1$ tuples
    beginning $010^p1^q\!$, $\beta\geq 1$ tuples beginning $100^p1^q$
    and $\gamma\geq 1$ tuples beginning $110^p1^q\!$, but none
    beginning $000^p1^q\!$.  We first show that, if $\alpha\neq
    \beta$, then we are done because $\Rneq\pppleq R$, so
    $R$ 3-simulates equality by Lemma~\ref{lemma:proj-eq}

    To this end, suppose $\alpha>\beta$ so, in particular, $\alpha\geq
    2$ and there are distinct tuples $010^p1^q\abar$ and
    $010^p1^q\bbar$ in $R$.  We may assume, without loss of
    generality, that $a_1\neq b_1$.  Since $\beta\geq 1$, there is at
    least one tuple $100^p1^q\cbar\in R$.  Suppose, now that we pin
    the $(2+p+q+1)$th column of $R$ to $c_1$.  $R$ cannot contain any
    tuple $110^p1^q\dbar$ with $d_1=c_1$ because it is not possible to
    pin more columns and still ppp-define $\Ror$.  But then $R$
    contains tuples beginning with each of $010^p1^qc_1$ and
    $100^p1^qc_1$ and none beginning $000^p1^qc_1$ or $110^p1^qc_1$,
    so $\Rneq\pppleq R$.  We similarly have $\Rneq\pppleq R$ if
    $\alpha<\beta$.  From this point, we may assume that
    $\alpha=\beta$.

    Similarly, either $\Rneq\pppleq R'\!$, so we are done,
    or $R'$ contains $\alpha'$ tuples beginning with each of
    $010^{p'}1^{q'}$ and $100^{p'}1^{q'}\!$, $\gamma'$ tuples
    beginning $000^{p'}1^{q'}$ and no tuples beginning
    $110^{p'}1^{q'}\!$.

    We now show how to simulate equality.  We can 3-simulate $\Reqk$
    by replacing the constraint $\Reqk(x_1, \dots, x_k)$ with the
    following set of constraints, modelled on $\xi_k$:
\begin{align*}
    \Xi_k = &\ \{R(x_i y_i 0^p 1^q *) \mid 1\leq i\leq k\} \\
            &\ \cup \{R'(y_i x_{i+1} 0^{p'} 1^{q'} *) \mid 1\leq i< k\}
                    \cup \{R'(y_k x_1 0^{p'} 1^{q'} *)\}\,,
\end{align*}
    where the $y_i$ are fresh variables and, as before, $*$ denotes a
    fresh tuple of variables for each constraint, of the appropriate
    length.  By Lemma~\ref{lemma:xi-sat}, an assignment $\sigma$
    satisfies $\Xi_k$ if, and only if, $\sigma(x_1) = \dots =
    \sigma(x_k) \neq \sigma(y_1) = \dots = \sigma(y_k)$.

    Further, there are $\alpha$ ways to satisfy the variables denoted
    by $*$ in each $R$ constraint and $\alpha'$ ways in each $R'$
    constraint.  Therefore, there are $(\alpha\alpha')^k$ satisfying
    assignments for $\Xi_k$ corresponding to each satisfying
    assignment for $\Reqk$ and we are done.

    Notice that our assumption that the ppp~definitions of $\Ror$ in
    $R$ and $\Rnand$ in $R'$ involve the identity permutation, pinning
    sequential columns to zero and one and projecting to the first two
    columns was made only for the notational convenience of referring to
    ``tuples beginning $010^p1^q$'' and so on.  This being the case,
    there is no requirement that $R$ and $R'$ be distinct, so the
    proof is complete.
\end{proof}

Note that there are relations, such as $\Reqk[3]$, that 2-simulate
equality, though we do not require this, here, so we omit the proof.


\section{Classifying relations}
\label{sec:Trichotomy}

We are now ready to prove that every Boolean relation $R$ is in \ORconj{}, in \NANDconj{} or 3-simulates equality. Given $r$-ary relations $R_0$ and $R_1$, we write $R_0+R_1$ for the relation $\{0\abar\mid \abar\in R_0\} \cup \{1\abar\mid \abar\in R_1\}$. The proof of the classification is by induction on the arity of $R$ and proceeds by decomposing $R$ as $R_0+R_1$.

Recall that a width-zero \ORconj{} (or, equivalently, \NANDconj{})
relation is a complete relation, possibly padded with some constant
columns.

\begin{lemma}
    Let $R_0, R_1\subseteq \ORconj$ have arity $r$ and width zero and
    let $R=R_0+R_1$.  Then, $R\in\ORconj$, $R\in\NANDconj$ or $R$
    3-simulates equality.
\end{lemma}
\begin{proof}
    We may assume that $R$ has no constant columns, since adding or
    removing them does not affect whether a relation is
    \ORconj{} or \NANDconj{} or whether it 3-simulates equality.

    For $i \in [2,r+1]$, let $R'_i = \proj{1,i} R$, so each $R'_i\pppleq
    R$.  If any $R'_i$ is $\Req$, $\Rneq$, $\Rimp$ or $\Rpmi$ then $R$
    3-simulates equality by Lemma~\ref{lemma:proj-eq}.  Otherwise,
    each $R'_i$ is either $\Bool^2\!$, $\Ror$ or $\Rnand$.
    If $R'_j=\Ror$ and $R'_k=\Rnand$ for some $j$ and $k$, then $R$
    3-simulates equality by Lemma~\ref{lemma:OR-NAND-eq}.  Otherwise,
    if no $R'_i=\Rnand$, let $I = \{i\mid R'_i=\Ror\}$.  Then,
    \begin{equation*}
        R = \bigwedge_{i\in I} \OR(x_1,x_i)\,,
    \end{equation*}
    so $R\in \ORconj$.  If no $R'_i=\Ror$, then $R\in\NANDconj$, by a
    similar argument.
\end{proof}

\begin{lemma}
\label{lemma:OR-OR}
    Let $R_0, R_1\subseteq \Bool^r$ be \ORconj{} and let $R=R_0+R_1$.
    Then, $R\in\ORconj$, $R\in\NANDconj$ or $R$ 3-simulates equality.
\end{lemma}
\begin{proof}
    We may assume, as before, that $R$ has no constant columns.  We
    may also assume that at least one of $R_0$ and $R_1$ has positive
    width: otherwise, the result is immediate from the previous lemma.
    We split the remaining work into two cases.

    \smallskip
    \noindent
    \emph{{\bf Case 1:} $R_0\subseteq R_1$.}
        Note that $R_1$ cannot have any constant columns in this case,
        since the same column would also have to be constant in $R_0$,
        giving a constant column in $R$.

        Suppose $R_i$ is defined by the normalized \ORconj{} formula
        $\phi_i$ in variables $x_2,\dots, x_{r+1}$.  Then $R$ is
        defined by the formula
        \begin{align}
            \phi_0 \vee (x_1=1 \wedge \phi_1)
            &\equiv (\phi_0 \vee x_1=1)
                \wedge (\phi_0 \vee \phi_1) \notag \\
            &\equiv (\phi_0 \vee x_1=1) \wedge \phi_1\,,
                \label{eq:OR-conj-formula}
        \end{align}
        where the first equivalence is the distribution law and the
        second is because $\phi_0$ implies $\phi_1$ (because
        $R_0\subseteq R_1$).  We consider the following two cases.

    \smallskip
    \noindent
    \emph{{\bf Case 1.1:} $R_0$ has no constant columns.}
        $\phi_0$ contains no pins and $x_1=1$ is equivalent to
        $\OR(x_1)$ so we can rewrite $\phi_0 \vee x_1=1$ in CNF.
        Therefore, (\ref{eq:OR-conj-formula}) defines an
        \ORconj{} relation.

    \smallskip
    \noindent
    \emph{{\bf Case 1.2:} $R_0$ has a constant column.}
        $R_1$ has no constant columns so, if $\proj{k} R_0 = \{0\}$
        for some $k$, then $\proj{1,k+1} R = \Rpmi$, and $R$
        3-simulates equality by Lemma~\ref{lemma:proj-eq}.  If every
        constant column of $R_0$ is all ones, then $\phi_0$ is in CNF
        since every pinning $x_i=1$ in $\phi_0$ can be written
        $\OR(x_i)$. We can therefore rewrite $\phi_0 \vee x_1=1$ in
        CNF, as in Case 1.1.

    \smallskip
    \noindent
    \emph{{\bf Case 2:} $R_0\nsubseteq R_1$.}  We will show that $R$
        3-simulates equality or is in \NANDconj{}. We consider two cases.

    \smallskip
    \noindent
    \emph{{\bf Case 2.1:} $R_1$ has a constant column,} say the $k$th.
        If the $k$th column of $R_0$ is also constant, it must have
        the opposite value (or $R$ would have a constant column).
        Therefore, $\proj{1,k+1} R$ is either $\Req$ or $\Rneq$, so
        $R$ 3-simulates equality by Lemma~\ref{lemma:proj-eq}.

        Otherwise, the $k$th column of $R_0$ is not constant, so
        $\proj{1,k+1} R$ is either $\Rimp$ or $\Rnand$.  In the first
        case, $R$ 3-simulates equality by Lemma~\ref{lemma:proj-eq}.
        In the second case, $\Ror$ is ppp-definable in at least one of
        $R_0$ and $R_1$ by Lemma~\ref{lemma:ORconj-OR} so $R$
        3-simulates equality by Lemma~\ref{lemma:OR-NAND-eq}.

    \smallskip
    \noindent
    \emph{{\bf Case 2.2:} $R_1$ has no constant columns.}
        By Proposition~\ref{prop:OR-monotone}, $R_1$ is monotone. Let
        $\abar\in R_0\setminus R_1$: by applying the same permutation
        to the columns of $R_0$ and $R_1$, we may assume that $\abar =
        0^\ell 1^{r-\ell}$.  We must have $\ell\geq 1$ as every
        non-empty $r$-ary monotone relation contains the tuple
        $1^r\!$.  Let $\bbar\in R_1$ be a tuple such that $a_i=b_i$
        for all $i$ in a maximal initial segment of $[1,r]$.  By monotonicity of
        $R_1$, we may assume that $\bbar = 0^k 1^{r-k}$.  Further, we
        must have $k<\ell$, since, otherwise, we would have
        $\bbar<\abar$, contradicting our choice of $\abar\notin R_1$.

        Now, consider the relation
        \begin{equation*}
            R' = \{a_0 a_1\dots a_{\ell-k}\mid
                            a_00^ka_1\dots a_{\ell-k}1^{r-\ell} \in R\}\,,
        \end{equation*}
        which is the result of pinning columns 2 to $(k+1)$ of $R$ to
        zero and columns $(r-\ell+1)$ to $(r+1)$ to one and discarding
        the resulting constant columns.  $R'$ contains $0^{\ell-k+1}$
        and $1^{\ell-k+1}$ but is not complete, as $10^{\ell-k}\notin
        R'\!$.  By Lemma~\ref{lemma:valid-eq}, $R'$ 3-simulates
        equality, so $R$ does, too.
\end{proof}

The following corollary follows from Proposition~\ref{prop:OR-NAND}
and the facts that $\bwcomp{R_0+R_1} = \bwcomp{R_1} + \bwcomp{R_0}$
and that, if $\bwcomp{R}$ 3-defines equality, then so does $R$, since
$\Req = \bwcomp{\Req}$.
 
\begin{corollary}
\label{cor:NAND-NAND}
    Let $R_0, R_1\in\NANDconj$ and let $R=R_0+R_1$. Then $R\in\ORconj$, $R\in\NANDconj$ or $R$ 3-simulates equality.
\end{corollary}

\begin{theorem}
\label{thrm:trichotomy}
    Every Boolean relation is in \ORconj{}, is in \NANDconj{} or
    3-simulates equality.
\end{theorem}
\begin{proof}
    Let $R$ be a Boolean relation.  We proceed by induction on its
    arity, $r$.  If $r\leq 2$, then, if $R$ is
    neither \ORconj{} nor \NANDconj{} then it can only be $\Req$,
    $\Rneq$, $\Rimp$ or $\Rpmi$; all of these 3-simulate equality by
    Lemma~\ref{lemma:proj-eq}.

    Now let $R$ have arity $r+1>2$ and
    let $R_0$ and $R_1$ be such that $R = R_0+R_1$.  By the inductive
    hypothesis, each of $R_0$ and $R_1$ is in \ORconj{}, in
    \NANDconj{} or 3-simulates equality.

    If either of $R_0$ and $R_1$ 3-simulates equality, then so does
    $R$.  Otherwise, either both are in \ORconj{}, both are in
    \NANDconj{} or one is in \ORconj{} and the other in \NANDconj{}.
    In the first two cases, $R$ is in \ORconj{} or in \NANDconj{} or
    3-simulates equality by Lemma~\ref{lemma:OR-OR} or
    Corollary~\ref{cor:NAND-NAND}.  In the third case, if $R_0$ and
    $R_1$ have positive width, then $R$ 3-simulates equality by
    Lemma~\ref{lemma:OR-NAND-eq}; otherwise, we are in one of the
    first two cases.
\end{proof}


\section{Complexity}
\label{sec:Complexity}

The complexity of approximating $\numCSP(\Gamma)$ where the degree of
instances is unbounded is given by Dyer, Goldberg and Jerrum
\cite[Theorem~3]{DGJ2010:Bool-approx}.

\begin{theorem}
\label{thrm:unbounded}
    Let $\Gamma$ be a Boolean constraint language.
    \begin{itemize}
    \itemspacing
    \item If every $R\in\Gamma$ is affine, then $\numCSP(\Gamma)\in
        \FPtime$.
    \item Otherwise, if $\Gamma\subseteq\IMconj$, then $\numCSP(\Gamma)
        \APequiv \numBIS$.
    \item Otherwise, $\numCSP(\Gamma) \APequiv \numSAT$.
    \end{itemize}
\end{theorem}

Towards our classification of the approximation complexity of
bounded-degree $\numCSP(\Gamma)$, we first deal with sub-cases.
Recall that $\numBIS$ is the problem of counting independent sets in
bipartite graphs and $\numwHISd$ is that of counting independent sets
in hypergraphs where every vertex has degree at most $d$ and every
hyper-edge contains at most $w$ vertices.

\begin{proposition}
\label{prop:im-bis}
    If $\Gamma \subseteq \IMconj$ contains at least one non-affine
    relation then, for all $d\geq 3$, $\numCSPd(\Gamma\cup \GammaPin)
    \APequiv \numBIS$.
\end{proposition}
\begin{proof}
    It is immediate from \cite[Lemma~9]{DGJ2010:Bool-approx} that $\numCSPd(\Gamma\cup \GammaPin) \APredto \numBIS$.

    For the converse, first observe that, by
    \cite[Lemma~8]{DGJ2010:Bool-approx}, $\numBIS\APredto
    \numCSP(\{\Rimp\})$ and, since $\Rimp$ 3-simulates equality by
    Lemma~\ref{lemma:proj-eq}, we have $\numCSP(\{\Rimp\}) \APredto
    \numCSPd(\{\Rimp\}\cup \GammaPin)$ for all $d\geq 3$ by
    Proposition~\ref{prop:bound-unbound}.  We must show that
    $\numCSPd(\{\Rimp\}\cup \GammaPin)\APredto \numCSPd(\Gamma\cup \GammaPin)$.

    To this end, let $R$ be any non-affine relation in $\Gamma$.  By
    Lemma~\ref{lemma:IMconj-implies}, $\Rimp\pppleq R$ and the
    ppp~definition involves projecting only pinned columns.
    Therefore, we can express the constraint $\Rimp(x,y)$ by a
    constraint of the form $R(v_1, \dots, v_r)$, where, for some $i$
    and $j$, $v_i = x$ and $v_j = y$ and the other variables are
    pinned to zero or one.
\end{proof}

\begin{lemma}
\label{lemma:OR-HIS}
    For $d\geq 2$ and $w\geq 2$,
    \begin{equation*}
        \numwHISd{} \APequiv \numCSPd(\{\Rork[w]\}\cup \GammaPin)
            \APequiv \numCSPd(\{\Rnandk[w]\}\cup \GammaPin).
    \end{equation*}
\end{lemma}
\begin{proof}
    The second equivalence is trivial, since $\Rork[w]$ and
    $\Rnandk[w]$ are bit-wise-complements.

    For the first equivalence, let $H$ be an instance of \numwHISd{}. We create an instance of $\numCSPd(\{\Rork[w]\}\cup \GammaPin)$ as follows.  The variables are $\{x_v\mid v\in V(H)\}$ and, for each hyper-edge $\{v_1, \dots, v_s\}$, there is a constraint $\Rork[w](x_{v_1}, \dots, x_{v_s}, 0, \dots, 0)$.  Each vertex appears in at most $d$ hyper-edges so each variable appears in at most $d$ constraints. It is easy to see that a configuration $\sigma$ of the resulting $\numCSPd(\{\Rork[w]\}\cup \GammaPin)$ instance is satisfying if, and only if, $\{v \mid \sigma(x_v) = 0\}$ is an independent set in $H$.

    Conversely, if we are given an instance of $\numCSPd(\{\Rork[w]\}\cup \GammaPin)$, we create an instance $H$ of \numwHISd{} as follows.  There is a vertex $v_x$ for every variable $x$.  For every constraint $\Rork[w](x_1, \dots, x_w)$ (where the $x_i$ are not necessarily distinct), add the hyper-edge $\{v_{x_1}, \dots, v_{x_w}\}$.  Now, for every constraint $\Rzero(x)$, delete the vertex $v_x$ and remove it from every hyper-edge that contains it.  For every constraint $\Rone(x)$, delete $v_x$ and delete every hyper-edge that contains it.  It is easy to see that a configuration $\sigma$ is satisfying if, and only if, it satisfies the pins and the set $\{v_x\mid \sigma(x) = 0\}\cap V(H)$ is independent in $H$.
\end{proof}

In the following two propositions, we just prove the \ORconj{} cases;
the \NANDconj{} cases are equivalent.

\begin{proposition}
\label{prop:HIS-to-ORconj}
    Let $R$ be an \ORconj{} or \NANDconj{} relation of width~$w>0$.
    Then, for $d\geq 2$, $\numwHISd{} \APredto \numCSPd(\{R\}\cup
    \GammaPin)$.
\end{proposition}
\begin{proof}
    By Lemma~\ref{lemma:ORconj-OR}, $\Rork[w]\pppleq R$ and the ppp~definition involves pinning and then projecting away all but $w$
    of the columns.  Thus, an $\Rork[w]$-constraint can be simulated by
    an $R$-constraint in which some elements of the scope are
    constants.  The result follows from Lemma~\ref{lemma:OR-HIS}.
\end{proof}

We define the \emph{variable rank} of an \ORconj{} or \NANDconj{}
relation $R$ to be $\vrank{R}$, the greatest number of times that any
variable appears in the (unique) normalized formula that defines $R$.
We similarly define the variable rank of an \ORconj{} or \NANDconj{}
constraint language to be the maximum variable rank of the relations
within it.

\begin{proposition}
\label{prop:ORconj-to-HIS}
    Let $R$ be an \ORconj{} or \NANDconj{} relation of width~$w>0$ and
    variable rank~$k$.  Then, for $d\geq 2$, $\numCSPd(\{R\}\cup
    \GammaPin) \APredto \numwHISd[kd]$.
\end{proposition}
\begin{proof}
    Given an instance $I$ of $\numCSPd(\{R\}\cup \GammaPin)$, we
    produce an instance $I'$ of the problem
    $\numCSP(\{\Rork[2],\dots,\Rork[w]\}\cup \GammaPin)$ with the same
    variables by replacing every $R$-constraint with the
    $\Rork[i]$-constraints and pins corresponding to the normalized
    formula that defines $R$. Clearly, $Z(I) = Z(I')$ but a variable
    that appeared $d$ times in $I$ appears up to $kd$ times in
    $I'\!$, so we have established that
    \begin{align*}
        \numCSPd(\{R\}\cup \GammaPin)
            &\APredto \numCSPd[kd](\{\Rork[2],\dots,\Rork[w]\}
                                                      \cup \GammaPin) \\
            &\APredto \numCSPd[kd](\{\Rork[w]\}\cup \GammaPin)\,,
    \end{align*}
    where the last reduction holds because, for any $s<w$, the constraint
    $\Rork[s](x_1, \dots, x_s)$ is equivalent to $\Rork[w](x_1, \dots,
    x_s, 0, \dots, 0)$. By Lemma~\ref{lemma:OR-HIS},
    $\numCSPd[kd](\{\Rork[w]\}\cup \GammaPin) \APequiv
    \numwHISd[kd]$.
\end{proof}

We now give the complexity of approximating $\numCSPd(\Gamma\cup \GammaPin)$ for $d\geq 3$.

\begin{theorem}
\label{theorem:complexity}
    Let $\Gamma$ be a Boolean constraint language and let $d\geq 3$.
    \begin{itemize}
    \itemspacing
    \item If every $R\in\Gamma$ is affine, then $\numCSPd(\Gamma\cup \GammaPin) \in
        \FPtime$.
    \item Otherwise, if $\Gamma\subseteq \IMconj$, then
        $\numCSPd(\Gamma\cup \GammaPin) \APequiv \numBIS$.
    \item Otherwise, if $\Gamma\subseteq \ORconj$ or $\Gamma\subseteq
        \NANDconj$, then $\numwHISd \APredto \numCSPd(\Gamma\cup
        \GammaPin) \APredto \numwHISd[kd]$, where $w=\width{\Gamma}$
        and $k=\vrank{\Gamma}$.
    \item Otherwise, $\numCSPd(\Gamma\cup \GammaPin) \APequiv \numSAT$.
    \end{itemize}
\end{theorem}
\begin{proof}
    The first three cases are immediate from
    Theorem~\ref{thrm:unbounded} and Propositions \ref{prop:im-bis},
    \ref{prop:HIS-to-ORconj} and~\ref{prop:ORconj-to-HIS}.  Note that
    $\Gamma\cup \GammaPin$ is affine if, and only if, $\Gamma$ is.

    For the remaining case, suppose that $\Gamma$ is not affine,
    $\Gamma\nsubseteq
    \IMconj$, $\Gamma\nsubseteq \ORconj$ and $\Gamma\nsubseteq
    \NANDconj$. Since $\Gamma\cup \GammaPin$ is neither affine nor a
    subset of \IMconj{}, we have $\numCSP(\Gamma\cup
    \GammaPin)\APequiv \numSAT$ by Theorem~\ref{thrm:unbounded} so, if
    we can show that $\Gamma$ $d$-simulates equality, then
    $\numCSPd(\Gamma\cup \GammaPin) \APequiv \numCSP(\Gamma\cup
    \GammaPin)$ by Proposition~\ref{prop:bound-unbound} and we are
    done.  If $\Gamma$ contains a relation $R$ that is neither
    \ORconj{} nor \NANDconj{}, then $R$ 3-simulates equality by
    Theorem~\ref{thrm:trichotomy}.  Otherwise, $\Gamma$ must contain
    distinct relations $R_1\in\ORconj$ and $R_2\in\NANDconj$ that are
    non-affine so have width at least two, so $\Gamma$ 3-simulates
    equality by Lemma~\ref{lemma:OR-NAND-eq}.
\end{proof}

Sly has shown that there can be no FPRAS for the
problem of counting independent sets in graphs of maximum degree at
least~6, unless $\NPtime=\RPtime$ \cite{Sly2010:IS-inapprox}.  Clearly,
if there is no FPRAS for counting independent sets in
such graphs, there can be no FPRAS for $\numwHISd[d]$ with $w\geq 2$
and $d\geq 6$. Further, since $\numSAT$ is complete for $\numP$ with
respect to AP-reducibility \cite{DGGJ2004:Approx}, $\numSAT$ cannot
have an FPRAS unless $\NPtime=\RPtime$. Thus, Theorem~\ref{thmcor:main} below is an immediate corollary of
Theorem~\ref{theorem:complexity}.

\begin{theorem}
\label{thmcor:main}
    Let $\Gamma$ be a Boolean constraint language and let $d\geq 6$.
    \begin{itemize}
    \itemspacing
    \item If every $R\in\Gamma$ is affine, then $\numCSPd(\Gamma\cup \GammaPin) \in \FPtime$.
    \item Otherwise, if $\Gamma\subseteq \IMconj$, then $\numCSPd(\Gamma\cup \GammaPin) \APequiv \numBIS$.
    \item Otherwise, there is no FPRAS for $\numCSPd(\Gamma\cup \GammaPin)$, unless $\NPtime=\RPtime$.
    \end{itemize}
\end{theorem}

Note that $\Gamma \cup \GammaPin$ is affine (respectively, in \ORconj{}
or in \NANDconj{}) if, and only if, $\Gamma$ is.  Therefore, the case
for large-degree instances ($d\geq 6$) corresponds exactly in
complexity to the unbounded case \cite{DGJ2010:Bool-approx}.

For lower degree bounds, the picture is more complex. To put Theorem~\ref{theorem:complexity} in context, summarize what is known about the approximability of $\numwHISd$ for various values of $d$ and $w$.

The case $d=1$ is clearly in \FPtime{} (Theorem~\ref{thrm:degree-1})
and so is the case $d=w=2$, which corresponds to counting independent
sets in graphs of maximum degree two. For $d=2$ and width $w\geq 3$,
Dyer and Greenhill have shown that there is an FPRAS for $\numwHISd$
\cite{DG2000:IS-Markov}. For $d=3$, they have shown that there is an
FPRAS if the width $w$ is at most~3. For larger width, the
approximability of $\numwHISd[3]$ is still not known. With the width
restricted to $w=2$ (ordinary graphs), Weitz has shown that, for degree
$d\in \{3,4,5\}$, there is a deterministic approximation scheme that
runs in polynomial time (a PTAS) \cite{Wei2006:IS-threshold}. This
extends a result of Luby and Vigoda, who gave an FPRAS for $d\leq 4$
\cite{LV1999:Convergence}.  For $d>5$, approximating $\numwHISd$
becomes considerably harder. Dyer, Frieze and Jerrum showed that, for
$d=6$, the Monte Carlo Markov chain technique is likely to fail, in
the sense that a certain class of Markov chains are provably slowly mixing
\cite{DFJ2002:IS-sparse}. They also showed that, for $d=25$, there can
be no polynomial-time algorithm for approximate counting, unless
$\NPtime=\RPtime$. As mentioned above, Sly has recently improved on
this, showing that there can be no FPRAS for $d\geq 6$ unless
$\NPtime=\RPtime$.  Table~\ref{tab:aHIS-complexity} summarizes the
results.

\begin{table}[t]
\centering\renewcommand{\arraystretch}{1.15}
\begin{tabular}{|c|c|l|}
    \hline
    Degree $d$ & Width $w$ & Approximability of $\numwHISd[d]$ \\
    \hline
    $1$
        & $\geq 2$
        & Exact counting in \FPtime \\
    $2$
        & $2$
        & Exact counting in \FPtime \\
    $2$
        & $\geq 3$
        & FPRAS~\cite{DG2000:IS-Markov} \\
    $3$
        & $2,3$
        & FPRAS~\cite{DG2000:IS-Markov} \\
    $3,4,5$
        & $2$
        & PTAS~\cite{Wei2006:IS-threshold} \\
    $\geq6$
        & $\geq 2$
        & No FPRAS unless $\NPtime=\RPtime$~\cite{Sly2010:IS-inapprox} \\
    \hline
\end{tabular}
\caption{A summary of known approximability of $\numwHISd[d]$. For values of $d$ and $w$ not covered by the table, the approximability is still unknown.}
\label{tab:aHIS-complexity}
\end{table}

Returning to bounded-degree \numCSP{}, the case $d=2$ seems to have a
rather different flavour to higher degree bounds.  This is
also the case for decision CSP --- recall that the complexity of
degree-$d$ $\CSP(\Gamma\cup\GammaPin)$ is the same as unbounded-degree
$\CSP(\Gamma\cup\GammaPin)$ for all $d\geq 3$
\cite{DF2003:bdeg-gensat}, while degree-2 $\CSP(\Gamma\cup\GammaPin)$
is often easier than the unbounded-degree case
\cite{DF2003:bdeg-gensat,Fed2001:Fanout} but there are still
constraint languages $\Gamma$ for which the complexity of degree-2
$\CSP(\Gamma\cup\GammaPin)$ is open.

Our key techniques for determining the complexity of $\numCSPd(\Gamma
\cup \GammaPin)$ for $d\geq 3$ are the 3-simulation of equality and
Theorem~\ref{thrm:trichotomy}, which says that every Boolean relation
is in \ORconj{}, in \NANDconj{} or 3-simulates equality.  However, it
seems that not all relations that 3-simulate equality also 2-simulate
equality so the corresponding classification of relations does not
appear to hold.  It seems that different techniques will be required
for the degree-2 case.  For example, it is possible that there is no
FPRAS for \numBIS{} and, therefore, no FPRAS for $\numCSPd[3](\Gamma
\cup \GammaPin)$ except when $\Gamma$ is affine.
However, Bubley and Dyer have shown that there is an FPRAS
for the restriction of \numSAT{} in which each variable appears at
most twice, even though the exact counting problem is \numP{}-complete
\cite{BD1997:Graph-orient}; the corresponding constraint language is
not affine.  This also shows that there is a
class $\mathcal{C}$ of constraint languages for which
$\numCSPd[2](\Gamma \cup \GammaPin)$ has an FPRAS for every $\Gamma
\in \mathcal{C}$ but for which no exact polynomial-time algorithm
exists, unless $\FPtime=\numP$.

We leave the complexity of degree-2 \numCSP{} and of \numBIS{} and
the various parameterized versions of the counting hypergraph
independent sets problem as open questions.


\bibliographystyle{plain}
\bibliography{bdegbool}

\end{document}